\documentclass[10pt,reqno]{amsart}
\usepackage{amsmath}
\usepackage{bbm}
\usepackage{amssymb}
\usepackage[utf8]{inputenc}
\newcommand*{\cA}{\mathcal{A}}
\newcommand*{\cB}{\mathcal{B}}
\newcommand*{\cC}{\mathcal{C}}
\newcommand*{\cH}{\mathcal{H}}
\newcommand*{\C}{\mathbb{C}}

\newcommand*{\N}{\mathbb{N}}
\newcommand*{\1}{\mathbf{1}}
\newcommand*{\R}{\mathbb{R}}
\newcommand*{\E}{\mathbb{E}}
\newcommand*{\halb}{\frac{1}{2}}
\newcommand{\abs}[1]{\left|#1\right|}
\newcommand{\norm}[1]{\left\|#1\right\|}

\renewcommand{\and}{\text{ and }}

\newcommand{\bmf}{\mathbf{f}}

\newtheorem{theorem}{Theorem}[section]
\newtheorem{lemma}[theorem]{Lemma}

\newtheorem{proposition}[theorem]{Proposition}
\newtheorem{corollary}[theorem]{Corollary}
\theoremstyle{definition}
\newtheorem{definition}[theorem]{Definition}
\newtheorem{example}[theorem]{Example}
\theoremstyle{remark}
\newtheorem{remark}[theorem]{Remark}
\numberwithin{equation}{section}
\allowdisplaybreaks

\begin{document}

%
%
%
%
%
%
%
%
%

\title[Generalized scaling operators]
 {Generalized Scaling Operators in White Noise Analysis and Applications to Ha\-mil\-tonian Path Integrals with Quadratic Action}

\author[Wolfgang Bock]{Wolfgang Bock}

\address{%
CMAF\\
Avenida Prof. Gama Pinto 2\\
1649-009\\
Portugal}

\email{bock@campus.ul.pt}

\thanks{This work is upported by the FCT-project: PTDC/MAT-STA/1284/2012}

\subjclass{Primary 99Z99; Secondary 00A00}

\keywords{Hamiltonian Path Integrals, White Noise Analysis, Scaling Operators}

\date{January 6, 2014}
\dedicatory{To the occasion of the 75th birthday of Ludwig Streit}

\begin{abstract}
We give an outlook, how to realize the ideas of complex scaling from \cite{GSV10}, \cite{GSV08}, \cite{GV08} to phase space path integrals in the framework of White Noise Analysis. The idea of this scaling method goes back to \cite{D80}. Therefore we extend the concept complex scaling to scaling with suitable bounded operators.
\end{abstract}

\maketitle


\section{Introduction}
As an alternative approach to quantum mechanics Feynman introduced the concept of path integrals (\cite{F48,Fe51,FeHi65}), which was developed into an extremely useful tool in many branches of theoretical physics. The phase space Feynman integral, or Hamiltonian path integral, for a particle moving from $y_0$ at time $0$ to $y$ at time $t$ under the potential $V$ is given by 
\begin{multline}\label{psfey} 
{\rm N} \int_{x(0)=y_0, x(t)=y} \int \exp\left(\frac{i}{\hbar} \int_0^t p\dot{x}-\frac{p^2}{2} -V(x,p) \, d\tau \right) \prod_{0<\tau<t}  dp(\tau) dx(\tau),\\
\hbar = \frac{h}{2\pi}.
\end{multline}
Here $h$ is Planck's constant, and the integral is thought of being over all position paths with $x(0)=y_0$ and $x(t)=y$ and all momentum paths. The missing restriction on the momentum variable at time $0$ and time $t$ is an immediate consequence of the Heisenberg uncertainty relation, i.e.~ the fact that one can not measure momentum and space variable at the same time. The path integral to the phase space has several advantages. Firstly the semi-classical approximation can be validated easier in a phase space formulation and secondly that quantum mechanics are founded on the phase space, i.e.~ every quantum mechanical observable can be expressed as a function of the space and momentum. A discussion about phase space path integrals can be found in the monograph \cite{AHKM08} and in the references therein.\\
There are many attempts to give a meaning to the Hamiltonian path integral as a mathematical rigorous object. Among these are analytic continuation of probabilistic integrals via coherent states \cite{KD82, KD84} and infinite dimensional distributions e.g.~\cite{DMN77}. Most recently also an approach using time-slicing was developed by Naoto Kumano-Go \cite{Ku11} and also by Albeverio et al.~using Fresnel integrals \cite{AHKM08, AGM02}. As a guide to the literature on many attempts to formulate these ideas we point out the list in \cite{AHKM08}. Here we choose a white noise approach.
White noise analysis is a mathematical framework which offers generalizations of concepts from finite-dimensional analysis, like differential operators and Fourier transform to an infinite-dimensional setting. We give a brief introduction to White Noise Analysis in Section 2, for more details see \cite{Hid80,BK95,HKPS93,Ob94,Kuo96}. Of special importance in White Noise Analysis are spaces of generalized functions and their characterizations. In this article we choose the space of Hida distributions, see Section 2.\\
The idea of realizing Feynman integrals within the white noise framework goes back to \cite{HS83}. There the authors used exponentials of quadratic (generalized) functions in order to give meaning to the Feynman integral in configuration space representation
\begin{equation*}
{\rm N}\int_{x(0) =y_0, x(t)=y} \exp\left(\frac{i}{\hbar} S(x) \right) \, \prod_{0<\tau<t} \, dx(\tau) ,\quad \hbar = \frac{h}{2\pi},
\end{equation*}
with the classical action $S(x)= \int_0^t \frac{1}{2} m \dot{x}^2 -V(x)\, d\tau$.
In \cite{BG11}, \cite{B13} and \cite{BG13} concepts of quadratic actions in White Noise Analysis, see \cite{GS98a} were used to give a rigorous meaning to the Feynman integrand 
\begin{align}\label{integrandpot}
I_V = {\rm Nexp}\left( \frac{i}{\hbar}\int_0^t  p(\tau) \dot{x}(\tau) -\frac{p(\tau)^2}{2m} d\tau +\frac{1}{2}\int_0^t \dot{x}(\tau)^2 +p(\tau)^2 d\tau\right)\\ \nonumber
\cdot \exp\left(-\frac{i}{\hbar} \int_0^t V(x(\tau),p(\tau),\tau) \, d\tau\right) \cdot \delta(x(t)-y)
\end{align}
as a Hida distribution. In this expression the sum of the first and the third integral in the exponential is the action $S(x,p)$, and the (Donsker's) delta function serves to pin trajectories to $y$ at time $t$. The second integral is introduced to simulate the Lebesgue integral by a local compensation of the fall-off of the Gaussian reference measure $\mu$.
Furthermore a Brownian motion starting in $y_0$ is used to model the space variable when the momentum variable is modeled by white noise, i.e.~
\begin{eqnarray}\label{varchoice}
x(\tau)=y(0)+\sqrt{\frac{\hbar}{m}}B(\tau),\quad 
p(\tau) =\omega(\tau),\quad 0\leq \tau \leq t. 
\end{eqnarray}
For the integrand we have thus the following ansatz 
\begin{multline*}
I_V={\rm Nexp}\big(-\frac{1}{2} \langle (\omega_x,\omega_p), K (\omega_x,\omega_p) \rangle \big) \cdot \exp\big(-\frac{1}{2} \langle (\omega_x,\omega_p), L (\omega_x,\omega_p) \rangle\big)\\
 \cdot \delta\big(\langle (\omega_x,\omega_p), (\1_{[0,t)},0) \rangle -y\big),
\end{multline*}
where $K$ is given by
\begin{equation}\label{kinmat}
K=\left(
\begin{array}[h]{l l}
    -\1_{[0,t)}&-i \1_{[0,t)}\\[0,1 cm]
    -i \1_{[0,t)}& -(1-i) \1_{[0,t)}
\end{array}
\right).
\end{equation}
Here the operator $\1_{[0,t)}$ denotes the multiplication with $\1_{[0,t)}$. And the operator $L$, fulfilling tha assumptions of Lemma \ref{thelemma} is used to model the quadratic potential.
For sake of simplicity we consider in this article path integrals with one degree of freedom, i.e.~the underlying space is the space $S'_{2}(\R)$.\\

In the euclidean configuration space a solution to the heat equation is given by the Feynman-Kac formula with its corresponding heat kernel. In White Noise Analysis one constructs the integral kernel by inserting Donsker's delta function to pin the final point $x\in \R$ and taking the expectation, i.e.,
$$
K_V(x,t,x_0,t_0)=\E\left(\exp(\int_{t_0}^t V(x_0+ \langle \1_{[t_0,r)},\cdot \rangle)\, dr) \delta(x_0+\langle \1_{[t_0,t)},\cdot \rangle-x)\right),$$
where the integrand is a suitable distribution in White Noise Analysis (e.g.~a Hida distribution).\\
A complementary strategy to construct Feynman intergals in the configuration space with White Noise methods was insprired by \cite{D80}, see also \cite{W95} and \cite{V10, GSV10} . Here for siutable potentials $V$ a complex-scaled Feynman-Kac kernel can be rigorously justified by giving a meaning to
\begin{eqnarray}\label{thecomplex}
K(x,t,x_0,t_0)=\E\left(\exp\left(\frac{1}{z^2} \int_{t_0}^t V(x+zB_s)\, ds\right) \sigma_z \delta(B_t -(x-x_0))\right).
\end{eqnarray}
In the configuration space, this has been done in \cite{W95} and \cite{V10, GSV10}. Note that if $z=\sqrt{i}$ in \eqref{thecomplex}, we have the Schrödinger kernel.\\
This scaling approach has several advantages e.g.
\begin{itemize}
\item Treatable potentials are beyond perturbation theory such as 
$$
V(x) = (-1)^{n+1} a_{4n+2} x^{4n+2} + \sum_{j=1}^{4n+1} a_j x^j,\quad x \in \R, n \in \N,$$ 
$\text{ with } a_{4n+2}>0, a_j \in \C.$

\item Due to a Wick formula we have a convenient structure (i.e.~"Brownian motion is replaced by a  Brownian bridge")
\item The kinetic energy $"\sigma_z \delta"$ and the potential can be treated separately, for details see e.g. \cite{LLSW94}.\\
\end{itemize}
We give an idea how to implement this approach to phase space for quadratic potentials. Thisis a first step to a scaling approach to the above potential class also for Hamiltonian path integrands.

\section{Preliminaries}

\subsection{Gel'fand Triples}
Starting point is the Gel'fand triple $S_d(\R) \subset L^2_d(\R) \subset S'_d(\R)$ of the $\R^d$-valued, $d \in \N$, Schwartz test functions and tempered distributions with the Hilbert space of (equivalence classes of) $\R^d$-valued square integrable functions w.r.t.~the Lebesgue measure as central space (equipped with its canonical inner product $(\cdot, \cdot)$ and norm $\|\cdot\|$), see e.g.~ \cite[Exam.~11]{W95}.
Since $S_d(\R)$ is a nuclear space, represented as projective limit of a decreasing chain of Hilbert spaces $(H_p)_{p\in \N}$, see e.g.~\cite[Chap.~2]{RS75a} and \cite{GV68}, i.e.~
\begin{equation*}
S_d(\R) = \bigcap_{p \in \N} H_p,
\end{equation*}
we have that $S_d(\R)$ is a countably Hilbert space in the sense of Gel'fand and Vilenkin \cite{GV68}. We denote the inner product and the corresponding norm on $H_p$ by $(\cdot,\cdot)_p$ and $\|\cdot\|_p$, respectively, with the convention $H_0 = L^2_d(\R)$.
Let $H_{-p}$ be the dual space of $H_p$ and let $\langle \cdot , \cdot \rangle$ denote the dual pairing on $H_{p} \times H_{-p}$. $H_{p}$ is continuously embedded into $L^2_d(\R)$. By identifying $L_d^2(\R)$ with its dual $L_d^2(\R)'$, via the Riesz isomorphism, we obtain the chain $H_p \subset L_d^2(\R) \subset H_{-p}$.
Note that $\displaystyle S'_d(\R)= \bigcup_{p\in \N} H_{-p}$, i.e.~$S'_d(\R)$ is the inductive limit of the increasing chain of Hilbert spaces $(H_{-p})_{p\in \N}$, see  e.g.~\cite{GV68}.
We denote the dual pairing of $S_d(\R)$ and $S'_d(\R)$ also by $\langle \cdot , \cdot \rangle$. Note that its restriction on $S_d(\R) \times L_d^2(\R)$ is given by $(\cdot, \cdot )$.
We also use the complexifications of these spaces denoted with the sub-index $\C$ (as well as their inner products and norms). The dual pairing we extend in a bilinear way. Hence we have the relation 
\begin{equation*}
\langle g,f \rangle = (\mathbf{g},\overline{\mathbf{f}}), \quad \mathbf{f},\mathbf{g} \in L_d^2(\R)_{\C},
\end{equation*}
where the overline denotes the complex conjugation.
\subsection{White Noise Spaces}
We consider on $S_d' (\R)$ the $\sigma$-algebra $\cC_{\sigma}(S_d' (\R))$ generated by the cylinder sets $\{ \omega \in S_d' (\R) | \langle \xi_1, \omega \rangle \in F_1, \dots ,\langle \xi_n, \omega \rangle \in F_n\} $, $\xi_i \in S_d(\R)$, $ F_i \in \cB(\R),\, 1\leq i \leq n,\, n\in \N$, where $\cB(\R)$ denotes the Borel $\sigma$-algebra on $\R$.\\
\noindent The canonical Gaussian measure $\mu$ on $C_{\sigma}(S_d'(\R))$ is given via its characteristic function
\begin{eqnarray*}
\int_{S_d' (\R)} \exp(i \langle {\bf f}, \boldsymbol{\omega} \rangle ) d\mu(\boldsymbol{\omega}) = \exp(- \tfrac{1}{2} \| {\bf f}\|^2 ), \;\;\; {\bf f} \in S_d(\R),
\end{eqnarray*}
\noindent by the theorem of Bochner and Minlos, see e.g.~\cite{Mi63}, \cite[Chap.~2 Theo.~1.~11]{BK95}. The space $(S_d'(\R),\cC_{\sigma}(S_d'(\R)), \mu)$ is the ba\-sic probability space in our setup.
The cen\-tral Gaussian spa\-ces in our frame\-work are the Hil\-bert spaces $(L^2):= L^2(S_d'(\R),$ $\cC_{\sigma}(S_d' (\R)),\mu)$ of complex-valued square in\-te\-grable func\-tions w.r.t.~the Gaussian measure $\mu$.\\
Within this formalism a representation of a d-dimensional Brownian motion is given by 
\begin{equation}\label{BrownianMotion}
{\bf B}_t ({\boldsymbol \omega}) :=(B_t(\omega_1), \dots, B_t(\omega_d)):= ( \langle  \1_{[0,t)},\omega_1 \rangle, \dots  \langle  \1_{[0,t)},\omega_d \rangle),\end{equation}
with ${\boldsymbol \omega}=(\omega_1,\dots, \omega_d) \in S'_d(\R),\quad t \geq 0,$
in the sense of an $(L^2)$-limit. Here $\1_A$ denotes the indicator function of a set $A$. 

\subsection{The Hida triple}

Let us now consider the Hilbert space $(L^2)$ and the corresponding Gel'fand triple
\begin{equation*}
(S) \subset (L^2) \subset (S)'.
\end{equation*}
Here $(S)$ denotes the space of Hida test functions and $(S)'$ the space of Hida distributions. In the following we denote the dual pairing between elements of $(S)$ and $(S)'$ by $\langle \! \langle \cdot , \cdot \rangle \!\rangle$. 
Instead of reproducing the construction of $(S)'$ here we give its characterization in terms of the $T$-transform.\\
\begin{definition}
We define the $T$-transform of $\Phi \in (S)'$ by
\begin{equation*}
T\Phi({\bf f}) := \langle\!\langle  \exp(i \langle {\bf f}, \cdot \rangle),\Phi \rangle\!\rangle, \quad  {\bf f}:= ({ f_1}, \dots ,{ f_d }) \in S_{d}(\R).
\end{equation*}
\end{definition}

\begin{remark}
\begin{itemize}
\item[(i)] Since $\exp(i \langle {\bf f},\cdot \rangle) \in (S)$ for all ${\bf f} \in S_d(\R)$, the $T$-transform of a Hida distribution is well-defined.
\item[(ii)] For ${\bf f} = 0$ the above expression yields $\langle\!\langle \Phi, 1 \rangle\!\rangle$, therefore $T\Phi(0)$ is called the generalized expectation of $\Phi \in (S)'$
\item[(iii)] Another important examples of Hida test functions are the so-called coherent states or Wick exponentials
$$
:\!\exp(\langle {\bf f}, \cdot \rangle)\!: \,= \exp(-\frac{1}{2} \langle  {\bf f},  {\bf f} \rangle ) \cdot \exp(\langle  {\bf f}, \cdot \rangle) ,\quad {\bf f} \in S_d(\R).
$$
\end{itemize}
\end{remark}

\noindent In order to characterize the space $(S)'$ by the $T$-transform we need the following definition.

\begin{definition}
A mapping $F:S_{d}(\R) \to \C$ is called a {\emph U-functional} if it satisfies the following conditions:
\begin{itemize}
\item[U1.] For all ${\bf{f, g}} \in S_{d}(\R)$ the mapping $\R \ni \lambda \mapsto F(\lambda {\bf f} +{\bf g} ) \in \C$ has an analytic continuation to $\lambda \in \C$ ({\bf{ray analyticity}}).
\item[U2.] There exist constants $0<C,D<\infty$ and a $p \in \N_0$ such that 
\begin{equation*}
|F(z{\bf f})|\leq C\exp(D|z|^2 \|{\bf f} \|_p^2), 
\end{equation*}
for all $z \in \C$ and ${\bf f} \in S_{d}(\R)$ ({\bf{growth condition}}).
\end{itemize}
\end{definition}

\noindent This is the basis of the following characterization theorem. For the proof we refer to \cite{PS91,Kon80,HKPS93,KLPSW96}.

\begin{theorem}\label{charthm}
A mapping $F:S_{d}(\R) \to \C$ is the $T$-transform of an element in $(S)'$ if and only if it is a U-functional.
\end{theorem}
Theorem \ref{charthm} enables us to discuss convergence of sequences of Hida distributions by considering the corresponding $T$-transforms, i.e.~ by considering convergence on the level of U-functionals. The following corollary is proved in \cite{PS91,HKPS93,KLPSW96}.

\begin{corollary}\label{seqcor}
Let $(\Phi_n)_{n\in \N}$ denote a sequence in $(S)'$ such that:
\begin{itemize}
\item[(i)] For all ${\bf f} \in S_{d}(\R)$, $((T\Phi_n)({\bf f}))_{n\in \N}$ is a Cauchy sequence in $\C$.
\item[(ii)] There exist constants $0<C,D<\infty$ such that for some $p \in \N_0$ one has 
\begin{equation*}
|(T\Phi_n)(z{\bf f })|\leq C\exp(D|z|^2\|{\bf f}\|_p^2)
\end{equation*}
for all ${\bf f} \in S_{d}(\R),\, z \in \C$, $n \in \N$.
\end{itemize}
Then $(\Phi_n)_{n\in \N}$ converges strongly in $(S)'$ to a unique Hida distribution.
\end{corollary}

\begin{example}[Vector valued white noise]
\noindent Let $\,{\bf{B}}(t)$, $t\geq 0$, be the $d$-di\-men\-sional Brow\-nian motion as in \eqref{BrownianMotion}. 
Consider $$\frac{{\bf{B}}(t+h,\boldsymbol{\omega}) - {\bf{B}}(t,\boldsymbol{\omega})}{h} = (\langle \frac{\1_{[t,t+h)}}{h} , \omega_1 \rangle , \dots (\langle \frac{\1_{[t,t+h)}}{h} , \omega_d \rangle),\quad h>0.$$ 
Then in the sense of Corollary \ref{seqcor} it exists
\begin{eqnarray*}
\langle {\boldsymbol\delta_t}, {\boldsymbol \omega} \rangle := (\langle \delta_t,\omega_1 \rangle, \dots ,\langle \delta_t,\omega_d \rangle):= \lim_{h\searrow 0} \frac{{\bf{B}}(t+h,\boldsymbol{\omega}) - {\bf{B}}(t,\boldsymbol{\omega})}{h}.
\end{eqnarray*}
Of course for the left derivative we get the same limit. Hence it is natural to call the generalized process $\langle {\boldsymbol\delta_t}, {\boldsymbol \omega} \rangle$, $t\geq0$ in $(S)'$ vector valued white noise. One also uses the notation ${\boldsymbol \omega}(t) =\langle{\boldsymbol\delta_t}, {\boldsymbol \omega} \rangle$, $t\geq 0$. 
\end{example}

Another useful corollary of Theorem \ref{charthm} concerns integration of a family of generalized functions, see \cite{PS91,HKPS93,KLPSW96}.

\begin{corollary}\label{intcor}
Let $(\Lambda, \cA, \nu)$ be a measure space and $\Lambda \ni\lambda \mapsto \Phi(\lambda) \in (S)'$ a mapping. We assume that its $T$--transform $T \Phi$ satisfies the following conditions:
\begin{enumerate}
\item[(i)] The mapping $\Lambda \ni \lambda \mapsto T(\Phi(\lambda))({\bf f})\in \C$ is measurable for all ${\bf f} \in S_d(\R)$.
\item[(ii)] There exists a $p \in \N_0$ and functions $D \in L^{\infty}(\Lambda, \nu)$ and $C \in L^1(\Lambda, \nu)$ such that 
\begin{equation*}
   \abs{T(\Phi(\lambda))(z{\bf f})} \leq C(\lambda)\exp(D(\lambda) \abs{z}^2 \norm{{\bf f}}^2), 
\end{equation*}
for a.e.~$ \lambda \in \Lambda$ and for all ${\bf f} \in S_d(\R)$, $z\in \C$.
\end{enumerate}
Then, in the sense of Bochner integration in $H_{-q} \subset (S)'$ for a suitable $q\in \N_0$, the integral of the family of Hida distributions is itself a Hida distribution, i.e.~$\!\displaystyle \int_{\Lambda} \Phi(\lambda) \, d\nu(\lambda) \in (S)'$ and the $T$--transform interchanges with integration, i.e.~
\begin{equation*}
   T\left( \int_{\Lambda} \Phi(\lambda)  d\nu(\lambda) \right)(\mathbf{f}) =
   	\int_{\Lambda} T(\Phi(\lambda))(\mathbf{f}) \, d\nu(\lambda), \quad \mathbf{f} \in S_d(\R).
\end{equation*}
\end{corollary}

Based on the above theorem, we introduce the following Hida distribution.
\begin{definition}
\label{D:Donsker} 
We define Donsker's delta at $x \in \R$ corresponding to $0 \neq {\boldsymbol\eta} \in L_{d}^2(\R)$ by
\begin{equation*}
   \delta_0(\langle {\boldsymbol\eta},\cdot \rangle-x) := 
   	\frac{1}{2\pi} \int_{\R} \exp(i \lambda (\langle {\boldsymbol\eta},\cdot \rangle -x)) \, d \lambda
\end{equation*}
in the sense of Bochner integration, see e.g.~\cite{HKPS93,LLSW94,W95}. Its $T$--transform in ${\bf f} \in S_d(\R)$ is given by
\begin{multline*}
   T(\delta_0(\langle  {\boldsymbol\eta},\cdot \rangle-x)({\bf f}) \\
   	= \frac{1}{\sqrt{2\pi \langle {\boldsymbol\eta}, {\boldsymbol\eta}\rangle}} \exp\left( -\frac{1}{2\langle {\boldsymbol\eta},{\boldsymbol\eta} \rangle}(i\langle {\boldsymbol\eta},{\bf f} \rangle - x)^2 -\frac{1}{2}\langle {\bf f},{\bf f}\rangle \right), \, \, \mathbf{f} \in S_d(\R).
\end{multline*}
\end{definition}

\subsection{Generalized Gauss Kernels}
Here we review a special class of Hida distributions which are defined by their T-transform, see e.g.~\cite{GS98a}. Let $\mathcal{B}$ be the set of all continuous bilinear mappings $B:S_{d}(\R) \times S_{d}(\R) \to \mathbb{C}$. Then the functions
\begin{equation*}
S_d(\R)\ni f \mapsto \exp\left(-\frac{1}{2} B({\bf f},{\bf f})\right) \in \mathbb{C}
\end{equation*}
for all $B\in \mathcal{B}$ are U-functionals. Therefore, by using the characterization of Hida distributions in Theorem \ref{charthm},
the inverse T-transform of these functions 
\begin{equation*}
\Phi_B:=T^{-1} \exp\left(-\frac{1}{2} B\right)
\end{equation*}
are elements of $(S)'$.

\begin{definition}\label{GGK}
The set of {\bf{generalized Gauss kernels}} is defined by
\begin{equation*}
GGK:= \{ \Phi_B,\; B\in \mathcal{B} \}.
\end{equation*}
\end{definition}

\begin{example}{\cite{GS98a}} \label{Grotex} We consider a symmetric trace class operator ${K}$ on $L^2_{d}(\R)$ such that $-\frac{1}{2}<{K}\leq 0$, then
\begin{align*}
\int_{S'_{d}(\R)} \exp\left(- \langle \omega,{K} \omega\rangle \right) \, d\mu(\boldsymbol{\omega}) 
= \left( \det({Id +2K})\right)^{-\frac{1}{2}} < \infty.
\end{align*}
For the definition of $\langle \cdot,{K} \cdot \rangle$ see the remark below.
Here ${Id}$ denotes the identity operator on the Hilbert space $L^2_{d}(\R)$, and $\det({A})$ of a symmetric trace class operator ${A}$ on $L^2_{d}(\R)$ denotes the infinite product of its eigenvalues, if it exists. In the present situation we have $\det({Id +2K})\neq 0$.
There\-fore we obtain that the exponential $g= \exp(-\frac{1}{2} \langle \cdot,{K} \cdot \rangle)$ is square-integrable and its T-transform is given by 
\begin{equation*}
Tg({\bf f}) = \left( \det({Id+K}) \right)^{-\frac{1}{2}} \exp\left(-\frac{1}{2} ({\bf f}, {(Id+K)^{-1}} {\bf f})\right), \quad {\bf f} \in S_{d}(\R).
\end{equation*}
Therefore $\left( \det({Id+K}) \right)^{\frac{1}{2}}g$ is a generalized Gauss kernel.
\end{example}

\begin{definition}
Let $K: L^2_{d,\mathbb{C}}(\R, dx) \to L^2_{d,\mathbb{C}}(\R, dx)$ be linear and continuous such that
\begin{itemize}
\item[(i)] $Id+K$ is injective, 
\item[(ii)] there exists $p \in \N_0$ such that $(Id+K)(L^2_{d,\mathbb{C}}(\R,\,dx)) \subset H_{p,\mathbb{C}}$ is dense,
\item[(iii)] there exist $q \in\N_0$ such that $(Id+K)^{-1} :H_{p,\mathbb{C}} \to H_{-q,\mathbb{C}}$ is continuous with $p$ as in (ii).
\end{itemize}
Then we define the normalized exponential
\begin{equation}\label{Nexp}
{\rm{Nexp}}(- \frac{1}{2} \langle \cdot ,K \cdot \rangle)
\end{equation}
by
\begin{align*}
T({\rm{Nexp}}(- \frac{1}{2} \langle \cdot ,K \cdot \rangle))({\bf f}) &:= \exp(-\frac{1}{2} \langle {\bf f}, (Id+K)^{-1} {\bf f} \rangle),\quad {\bf f} \in S_d(\R).
\end{align*}
\end{definition}

\begin{remark}
The "normalization" of the exponential in the above definition can be regarded as a division of a divergent factor. In an informal way one can write
\begin{multline*}
T({\rm{Nexp}}(- \frac{1}{2} \langle \cdot ,K \cdot \rangle))({\bf f})=\frac{T(\exp(- \frac{1}{2} \langle \cdot ,K \cdot \rangle))(\bf f)}{T(\exp(- \frac{1}{2} \langle \cdot ,K \cdot \rangle))(0)}\\
=\frac{T(\exp(- \frac{1}{2} \langle \cdot ,K \cdot \rangle))(\bf f)}{\sqrt{\det(Id+K)}} , \quad {\bf f} \in S_d(\R), 
\end{multline*}
i.e.~we can still define the normalized exponential by the T-transform even if the determinant is not defined.
\end{remark}

\begin{lemma}{\cite{BG11}}\label{thelemma}
Let  ${L}$ be a $d\times d$ block operator matrix on $L^2_{d}(\R)_{\C}$ acting component-wise such that all entries are bounded operators on $L^2(\R)_{\C}$.
Let ${K}$ be a d $\times d$ block operator matrix on $L^2_{d}(\R)_{\C}$, such that ${Id+K}$ and ${N}={Id}+{K}+{L}$ are bounded with bounded inverse. Furthermore assume that $\det({Id}+{L}({Id}+{K})^{-1})$ exists and is different from zero (this is e.g.~the case if ${L}$ is trace class and -1 in the resolvent set of ${L}({Id}+{K})^{-1}$).
Let $M_{{N}^{-1}}$ be the matrix given by an orthogonal system $({\boldsymbol\eta}_k)_{k=1,\dots J}$ of non--zero functions from $L^2_d(\R)$, $J\in \N$, under the bilinear form $\left( \cdot ,{N}^{-1} \cdot \right)$, i.e.~ $(M_{{N}^{-1}})_{i,j} = \left( {\boldsymbol\eta}_i ,{N}^{-1} {\boldsymbol\eta}_j \right)$, $1\leq i,j \leq J$.
Under the assumption that either 
\begin{eqnarray*}
\Re(M_{\mathbf{N}^{-1}}) >0 \quad \text{ or }\quad \Re(M_{\mathbf{N}^{-1}})=0 \,\text{ and } \,\Im(M_{\mathbf{N}^{-1}}) \neq 0,
\end{eqnarray*} 
where $M_{\mathbf{N}^{-1}}=\Re(M_{\mathbf{N}^{-1}}) + i \Im(M_{\mathbf{N}^{-1}})$ with real matrices $\Re(M_{\mathbf{N}^{-1}})$ and $\Im(M_{\mathbf{N}^{-1}})$, \\
then
\begin{equation*}
\Phi_{{K},{L}}:={\rm Nexp}\big(-\frac{1}{2} \langle \cdot, {K} \cdot \rangle \big) \cdot \exp\big(-\frac{1}{2} \langle \cdot, {L} \cdot \rangle \big) \cdot \exp(i \langle \cdot, {\bf g} \rangle)
\cdot \prod_{i=1}^J \delta_0 (\langle \cdot, {\boldsymbol\eta}_k \rangle-y_k),
\end{equation*}
for ${\bf g} \in L^2_{d}(\R,\C),\, t>0,\, y_k \in \R,\, k =1\dots,J$, exists as a Hida distribution. \\
Moreover for ${\bf f} \in S_d(\R)$
\begin{multline}\label{magicformula}
T\Phi_{K,{L}}({\bf f})=\frac{1}{\sqrt{(2\pi)^J  \det((M_{{N}^{-1}}))}}
\sqrt{\frac{1}{\det({Id}+{L}({Id}+{K})^{-1})}}\\ 
\times \exp\bigg(-\frac{1}{2} \big(({\bf f}+{\bf g}), {N}^{-1} ({\bf f}+{\bf g})\big) \bigg)
\exp\bigg(-\frac{1}{2} (u,(M_{{N}^{-1}})^{-1} u)\bigg),
\end{multline}
where
\begin{equation*}
u= \left( \big(iy_1 +({\boldsymbol\eta}_1,{N}^{-1}({\bf f}+{\bf g})) \big), \dots, \big(iy_J +({\boldsymbol\eta}_J,{N}^{-1}({\bf f}+{\bf g})) \big) \right).
\end{equation*}
\end{lemma}

\subsection{Scaling Operator}
First note, that every test function $\varphi \in (S)$ can be extended to $S_d(\R)'_{\C}$, see e.g.~\cite{Kuo96}. Thus the following definition makes sense. 
\begin{definition}
Let $\varphi$ be the continuous version of an element of $(S)$. Then for $0\neq z \in \C$ we define the scaling operator $\sigma_z$ by
$$
(\sigma_z \varphi)(\omega) = \varphi(z \omega), \quad \omega \in S_d'(\R).$$ 
\end{definition} 

\begin{proposition}
\begin{enumerate}
\item[(i)] For all $0\neq z \in \C$ we have $\sigma_z \in L((S),(S))$,
\item[(ii)] for $\varphi, \psi\in (S)$ we have $$\sigma_z (\varphi \cdot \psi) = (\sigma_z \varphi) (\sigma_z \psi).$$
\end{enumerate}
\end{proposition}
\begin{proof}
$(i)$ is proved in \cite{W95}\\
For $(ii)$, first note that $(S)$ is an algebra under pointwise multiplication. Since the scaling operator is continuous from $(S)$ to itself by $(i)$, it suffices to show the assumption for the set of Wick ordered exponentials. Since this set is total the rest follows by a density argument. 
We have for $\xi, \eta \in S_d(\R)$,

\begin{multline*} 
\sigma_z (:\exp(\langle \xi, \cdot\rangle): \cdot :\exp(\langle \eta, \cdot\rangle):)=\\ \sigma_z( \exp\left(-\frac{1}{2} (\langle \xi, \xi \rangle + \langle \eta, \eta \rangle )\right) \exp(\langle \xi+\eta, \omega\rangle)\\
= \exp\left(-\frac{1}{2} (\langle \xi, \xi \rangle + \langle \eta, \eta \rangle )\right) \exp(\langle \xi +\eta, z \omega\rangle)\\
=  \exp\left(-\frac{1}{2} (\langle \xi, \xi \rangle + \langle \eta, \eta \rangle )\right)\exp(\langle z \xi, \omega \rangle) \exp(\langle z \eta, \omega \rangle )
\end{multline*}
on the other hand
\begin{multline*}
\sigma_z (:\exp(\langle \xi, \cdot):) \cdot \sigma_z(:\exp(\langle \eta, \cdot):)) =\\ \exp(-\frac{1}{2} \langle \xi, \xi \rangle)\exp(\langle  \xi, z\omega \rangle) \exp(-\frac{1}{2} \langle \eta, \eta \rangle)\exp(\langle \eta, z\omega \rangle ),
\end{multline*}
which proves the assumption.
\end{proof}
More precisely we have, compare to \cite{V10} and \cite{W95} the following proposition.
\begin{proposition}
Let $\varphi \in (S)$, $z\in \C$, then 
$$\sigma_z\varphi = \sum_{n=0}^{\infty} \langle \varphi^{(n)}_{z} , :\cdot^{n}: \rangle,$$
with kernels 
$$
\varphi^{(n)}_{z} = z^n \sum_{k=0}^{\infty} \frac{(n+2k)!}{k! n!} \left(\frac{z^2 -1}{2}\right)^k \cdot tr^k \varphi^{(n+2k)}.
$$
\end{proposition}

\begin{definition}
Since $\sigma_z$ is a continuous mapping from $(S)$ to itself we can define its dual operator $\sigma_z^{\dag}: (S)^* \to (S)^*$ by
$$
\langle \! \langle \varphi,\sigma_z^{\dag} \Phi\rangle\! \rangle = \langle \! \langle \sigma_z\varphi, \Phi\rangle\! \rangle,
$$
for $\Phi \in  (S)^*$ and $\varphi \in  (S)$.
\end{definition}
The following proposition can be found in \cite{W95} and \cite{V10}.
\begin{proposition}
Let $\Phi \in (S)^{*}$, $\varphi,\psi \in (S)$ and $z \in \C$ then we have
\begin{itemize}
\item[(i)] $$\sigma_{z}^{\dag}\Phi = J_z \diamond \Gamma_z \Phi,$$
where $\Gamma_z$ is defined by 
$$S(\Gamma_z \Phi)(\xi) = S(\Phi)(z \xi),\quad \xi \in S_d(\R),$$
and $J_z= \mathrm{Nexp}(-\frac{1}{2} z^2 \langle \cdot , \cdot \rangle)$. In particular we have
$$
\sigma_{z}^{\dag}\1 = J_z.$$
\item[(ii)] $J_z\varphi =\sigma_z^{\dag}(\sigma_z \varphi).$
\end{itemize}
\end{proposition}

\section{Generalized Scaling Operators}
In a view to the previous section we want to generalize the notion of scaling to bounded operators. More precisely we investigate for which kind of linear mappings $B \in L(S(\R)', S(\R)')$ there exists some operator $\sigma_{B}:(\mathcal N)\to(\mathcal N)$ such that
\begin{align*}	
\Phi_{(BB^*)}\cdot\varphi :=\sigma_{B}^\dag\sigma_{B}\varphi.
\end{align*}
Further we state a generalization of the Wick formula to Gauss kernels.
We start with the definition of $\sigma_{B}$.

\begin{definition}\label{def_tr_B}
Let $B\in L(S_d(\R)_\C,S'_d(\R)_\C)$. By $tr_B$ we denote the element in $S'_d(\R)_\C\otimes S'_d(\R)_\C$, which is defined by
$$\forall \xi,\ \eta\in S_d(\R)_\C:\ tr_B(\xi\otimes\eta):= \left\langle \xi,B\eta\right\rangle.$$
Note that $tr_B$ is not symmetric. Further there exists a $q\in\mathbb{Z}$ such that $tr_B\in \cH_{q,\C}\otimes \cH_{q,\C}$.
\end{definition}
\begin{proposition}\label{p:trBN}
Let $B\in L(\cH_\C,\cH_\C)$ be a Hilbert-Schmidt operator. Then $tr_B \in \cH_\C\otimes\cH_\C$. Further for each orthonormal basis $(e_j)_{j\in\N}$ of $\cH_\C$ it follows:
$$tr_B = \sum\limits_{i=0}^{\infty} Be_i\otimes e_i.$$
\end{proposition}
\begin{proof}
$\left|\sum\limits_{i,j=0}^{\infty}\left\langle e_i,Be_j\right\rangle e_i\otimes e_j\right|_0^2 = \sum\limits_{i,j=0}^{\infty}\left|\left\langle e_i,Be_j\right\rangle\right|^2=\sum\limits_{j=0}^{\infty}\left|Be_j\right|_0^2=\left\|B\right\|_{HS}^2<\infty$.
Hence the sum is a well-defined element in $\cH_\C\otimes\cH_\C$. The identity follows by verifying the formula for $\left\{e_k\otimes e_l\right\}_{k,l\in\N}$:
$$
\left\langle e_k\otimes e_l,  \sum\limits_{i=0}^{\infty} Be_i\otimes e_i \right\rangle =(e_k, \overline{B e_l})_{\cH} (e_l,\overline{e_l})_{\cH} =(e_k, \overline{B e_l})_{\cH}  = \langle e_k, B e_l \rangle.$$
\end{proof}
\begin{proposition}
In the case $S(\R)=S(\R)$ and $B\in L(S_{\C}(\R),S'_{\C}(\R))$ we have
$$tr_B = \sum\limits_{i=0}^{\infty} Bh_i\otimes h_i$$
\end{proposition}

\begin{proof}
By the continuity of the bilinear form $\langle \cdot, B \cdot \rangle$ on $S(\R)_{\C}\times S(\R)_{\C}$ there exists $p\geq 0$ such that $B \in L(\cH_{p,\C},\cH_{-p,\C})$. Let $q>p+1$. Then
$$
|\sum_{n=0}^{\infty} Bh_n \otimes h_n |_{-q}^2 = \sum_{n=0}^{\infty} |Bh_n|_{-q}^2 \cdot |h_n|^2_{-q}
\leq K \sum_{n=0}^{\infty} |h_n|_p^2 \cdot |h_n|_{-q}^2,
$$
for some $K>0$. For the last expression we have 
$$
 K \sum_{n=0}^{\infty} |h_n|_p^2 \cdot |h_n|_{-q}^2
\leq K \sum_{n=0}^{\infty} \left(\frac{1}{2n+2}\right)^{2} <\infty.$$
Then as in the proof of Proposition \ref{p:trBN} we obtain
$$
tr_B = \sum_{n=0}^{\infty} Bh_n \otimes h_n,$$
since $$\overline{\mathrm{span} \{h_n \otimes h_l\}_{n,l}}^{S(\R)\otimes S(\R)} =S(\R)\otimes S(\R).$$  
\end{proof}
\begin{definition}
For $B\in L(S'_d(\R)_\C,S'_d(\R)_\C)$ we define $\sigma_{B}\varphi$, $\varphi\in (S)$, via its chaos decomposition, which is given by
\begin{align}\label{chaosscalt0t}
\sigma_{B}\varphi=\sum_{n=0}^\infty \left\langle {\varphi}_{B}^{(n)}, :\cdot^{\otimes n}: \right\rangle,
\end{align}
with kernels
\begin{align*}
		{\varphi}_{B}^{(n)}= \sum_{k=0}^\infty \frac{(n+2k)!}{k!n!}\left(-\halb\right)^k (B^*)^{\otimes n}({\rm{tr}}_{(Id-BB^*)}^k\varphi^{(n+2k)}).
\end{align*}
Here, $B^*$ means the dual operator of B with respect to $\left\langle \cdot,\cdot\right\rangle$. Further for $A\in L(S_d(\R)_\C, S'_d(\R)_\C)$, the expression  ${\rm{tr}}_{A}^k\varphi^{(n+2k)}$ is defined by
\begin{align*}	{\rm{tr}}_{A}^k\varphi^{(n+2k)}:=\left\langle {\rm{tr}}_{A}^{\otimes k},\varphi^{(n+2k)}\right\rangle\in {\mathcal N}^{\hat{\otimes} n},
\end{align*}
where the generalized trace kernel ${\rm{tr}}_{A}$ is defined in \ref{def_tr_B}.
\end{definition}

\begin{proposition}\label{tensor_estim}
Let $B\in L(S_d(\R)_\C, S_d(\R)_\C)$ and $n\in\N,\ n>0$. Then
for all $p\in\N$ there exists a $K>0$ and $q_1,q_2\in\N$, with $p<q_1<q_2$ such that for all $\theta\in S_d(\R)_\C^{\otimes n}$ we have
$$ \left|B^{\otimes n}\theta\right|_p \leq \left(K\left\|B\right\|_{q_2 , q_1}\left\|i_{q_1 , p}\right\|_{HS}\right)^{n}\cdot \left|\theta\right|_{q_2}.$$  
\end{proposition}
\begin{proof}
Choose $q_1,\ q_2\in\N$ with $q_2>q_1> p$ and
$$B\in L(H_{q_1,\C},H_{p,\C})\text{ and }B\in L(H_{q_2,\C},H_{q_1,\C},).$$
For a shorter notation we write $e_J$ for $e_{j_1}\otimes\cdots\otimes e_{j_n}$.
Now let $(e_J)_J$ be an orthonormal basis of $H_{p}^{\otimes n}$. Further let $\mathcal{I}_p:\ (H_{p,\C}^{\otimes n})'\rightarrow H_{p,\C}^{\otimes n}$ be the Riesz isomorphism. Then for $\theta\in S_d(\R)_\C^{\otimes n}$
\begin{align*}
\left|B^{\otimes n}\theta\right|_p^2 
&\ = \sum\limits_J \left|(B^{\otimes n}\theta,e_J)_p\right|^2 = \sum\limits_J \left|\left\langle B^{\otimes n} \theta,\overline{\mathcal{I}_p^{-1}(e_J)}\right\rangle\right|^2\\
&\ = \sum\limits_J \left|\left\langle B^{\otimes n}\theta,\mathcal{I}_p^{-1}(e_J)\right\rangle\right|^2 = \sum\limits_J \left|\left\langle \theta,(B^*)^{\otimes n}\mathcal{I}_p^{-1}(e_J)\right\rangle\right|^2\\
&\ \leq \left|\theta\right|_{q_2}^2\cdot \sum\limits_{J}\left|(B^*)^{\otimes n}\mathcal{I}_p^{-1}(e_J)\right|^2_{-q_2},
\end{align*}
We can estimate this with the norm of $B^*:H_{-q_1,\C} \to H_{-q_2,\C}$, denoted by $\|B^*\|_{-q_1 , -q_2}$. Then for a $K>0$ we have 
\begin{align*}
\left|\theta\right|_{q_2}^2\cdot \sum\limits_{J}\left|(B^*)^{\otimes n}\mathcal{I}_p^{-1}(e_J)\right|^2_{-q_2}&\ \leq \left|\theta\right|_{q_2}^2\cdot K^{2n} \left\|B^*\right\|_{-q_1 , -q_2}^{2n}\sum\limits_{J}\left|\mathcal{I}_p^{-1}(e_J)\right|^2_{-q_1},\\
&\ = \left|\theta\right|_{q_2}^2\cdot K^{2n} \left\|B\right\|_{q_1 , q_2}^{2n}\left\|i_{q_1,p}\right\|_{HS}^{2n},\\
\end{align*}
where the last equation is due to \cite[Theorem 4.10(2), p. 93]{Ru74}
\end{proof}
Next we show the continuity of the generalized scaling operator.

\begin{proposition}\label{prop:contgensca}
Let $B:S_d'(\R)_{\C} \to S_d'(\R)_{\C}$ a bounded operator. 
Then $\varphi\mapsto \sigma_{B}\varphi$ is continuous from $(S)$ into itself.
\end{proposition}    
\begin{proof}
Let ${\varphi}_{B}^{(n)}$ as in Definition \ref{chaosscalt0t}.\\
 First choose $q_1>0$, such that $|tr_{Id-BB^*}|^k_{-q_1}<\infty$. Then, by \ref{tensor_estim}, there exist $C(B)>0,\ q_2> q1$ such that
\begin{multline*}
|{\varphi}_{B}^{(n)}|_{p} = |\sum_{k=0}^\infty \frac{(n+2k)!}{k!n!}\left(-\halb\right)^k (B^*)^{\otimes n} ({\rm{tr}}_{Id-BB^*}^k\varphi^{(n+2k)})|_{p}\\
\leq \frac{1}{\sqrt{n!}} C(B)^n\sum_{k=0}^{\infty} \sqrt{\binom{n+2k}{2k}} \frac{\sqrt{(2k)!}}{k!2^k} \sqrt{(n+2k)!}  |tr_{Id-BB^*}|^k_{-q_2}|\varphi^{n+2k}|_{q_2}
\end{multline*}
Since   $\frac{\sqrt{(2k)!}}{k!2^k}<1$, see \cite{W95},  we have
\begin{multline*}
\frac{1}{\sqrt{n!}}  C(B)^n \sum_{k=0}^{\infty} \sqrt{\binom{n+2k}{2k}} \frac{\sqrt{(2k)!}}{k!2^k} \sqrt{(n+2k)!} |tr_{Id-BB^*}|^k_{-q_2} |\varphi^{n+2k}|_{q_2}\\
\leq \frac{1}{\sqrt{n!}} C(B)^n \sum_{k=0}^{\infty} \sqrt{\binom{n+k}{k}} \sqrt{(n+k)!} |tr_{Id-BB^*}|^{\frac{k}{2}}_{-q_2}  |\varphi^{n+k}|_{q_2} \\
\leq \frac{1}{\sqrt{n!}} C(B)^n 2^{-n \frac{q'}{2}} \left(\sum_{k=0}^{\infty} \binom{n+k}{k} 2^{-q'k}  |tr_{Id-BB^*}|^k_{-q_2}\right)^{\frac{1}{2}}\\
\times \left(\sum_{k=0}^{\infty} (n+k)! 2^{q'(n+k)} |\varphi^{(n+k)}|_{q_2}^2\right)^{\frac{1}{2}}\\
\leq \|\varphi\|_{q_2,q'} \frac{1}{\sqrt{n!}} 2^{-n \frac{q'}{2}} C(B)^n\left(1-2^{-q'}  |tr_{Id-BB^*}|_{-q_2}\right)^{-\frac{n+1}{2}}.
\end{multline*}
If $q'$ fulfills $$2^{-q'}  |tr_{Id-BB^*}|_{-q_2}<1,$$
we obtain
$$\|\sigma_B \varphi\|^2_{q_2,q} \leq \|\varphi\|_{q_2,q'}^2 \cdot \sum_{n=0}^{\infty} 2^{n(q-q')} C(B)^{2n} \left(1- 2^{-q'}   |tr_{Id-BB^*}|{-q_2}\right)^{-(n+1)},$$
where the right hand side converges if $q'-q$ is large enough.
\end{proof}
\begin{proposition}
Let $\varphi \in (S)$ given by its continuous version. Then it holds
$$\sigma_{B} \varphi (\omega) = \varphi(B \omega),$$
if $B \in L(S_d'(\R), S_d'(\R)),\ \omega\in S_d'(\R)$.\\
\end{proposition}
This can be proved directly by an explicit calculation on the the set of Wick exponentials, a density argument and a verifying of pointwise convergence, compare \cite[Proposition 4.6.7, p. 104]{Ob94}, last paragraph.\\
In  the same manner the following statement is proved.
\begin{proposition}
Let $B:S(\R)' \to S(\R)'$ be a bounded operator. 
For $\varphi,\psi\in(\mathcal N)$ the following equation holds
\begin{align*}
     \sigma_{B}(\varphi\psi)=(\sigma_{B}\varphi)(\sigma_{B}\psi).
 \end{align*}
\end{proposition}

Since we consider a continuous mapping from $(\mathcal N)$ into itself one
can define the dual scaling operator with respect to $\left\langle \cdot,\cdot\right\rangle$,
$\sigma_{B}^\dag:(\mathcal N)^{\prime}\to(\mathcal N)^{\prime}$ by
\begin{align*}
    \Big\langle \!\Big\langle\sigma_{B}^\dag\Phi, \psi\Big\rangle \! \Big\rangle= \Big\langle \! \Big\langle\Phi,\sigma_{B} \psi\Big\rangle
    \! \Big\rangle,\quad
\end{align*}

The Wick formula as stated in \cite{V10, GSV10} for Donsker's delta function can be extended to Generalized Gauss kernels.
\begin{proposition}\label{genwickform}[Generalized Wick formula]
Let $\Phi \in (S)'$, $\varphi,\psi \in (S)$ and $B\in L(S'_d(\R),S'_d(\R))$. then we have
\begin{itemize}
\item[(i)] $$\sigma_{B}^{\dag} = \Phi_{BB^*} \diamond \Gamma_B{^*} \Phi,$$
where $\Gamma_B{^*}$ is defined by 
$$S(\Gamma_B{^*} \Phi)(\xi) = S(\Phi)(B{^*} \xi),\quad \xi \in S_d(\R).$$ 
In particular we have
$$
\sigma_{B}^{\dag}\1 = \Phi_{BB^*}.$$
\item[(ii)] $\Phi_{BB^*}\cdot\varphi =\sigma_B^{\dag}(\sigma_B \varphi).$
\item[(iii)] $\Phi_{BB^*}\cdot\varphi =\Phi_{BB^*}\diamond (\Gamma_{B^*}\circ\sigma_B (\varphi)).$
\end{itemize}
\end{proposition}

\begin{proof}
Proof of (i): Let $\Phi \in (S)'$ and $\xi \in S_d(\R)_{\C}$ then we have
\begin{multline*}
S(\sigma_B^{\dag} \Phi)(\xi) = \langle \! \langle :\exp(\langle \xi, \cdot \rangle):, \sigma_B^{\dag} \Phi \rangle \! \rangle \\
= \langle \! \langle \sigma_B :\exp(\langle \xi, \cdot \rangle):, \Phi \rangle \! \rangle= \exp(-\frac{1}{2} \langle \xi, \xi \rangle ) \langle \! \langle \exp(\langle B^*\xi, \cdot \rangle), \Phi \rangle \! \rangle\\
=\exp(-\frac{1}{2} \langle \xi, (Id -BB^*) \xi \rangle) S(\Phi)(B^*\xi)
=S(\Phi_{BB^*})(\xi) \cdot S(\Gamma_{B^*} \Phi)(\xi)
\end{multline*}
Proof of (ii): 
First we have $\sigma_B^{\dag} \1 = \Phi_{BB^*} \diamond \Gamma_{B^*} \1 = \Phi_{BB^*}$.\\
Thus for all $\varphi, \psi \in (S)$\\
\begin{multline*}
\langle \! \langle \Phi_{BB^*}\varphi, \psi \rangle \! \rangle = \langle \! \langle \sigma^{\dag} \1,\varphi\cdot\psi \rangle \! \rangle =\\
 \langle \! \langle \1, (\sigma_B \varphi) (\sigma_B \psi)\rangle \! \rangle = \langle \! \langle (\sigma_B \varphi), (\sigma_B \psi)\rangle \! \rangle = \langle \! \langle \sigma_B^{\dag} (\sigma_B \varphi),\psi \rangle \! \rangle
\end{multline*}
Proof of (iii):
Immediate from (i) and (ii).
\end{proof}

\begin{remark}
The scaling operator can be considered as a linear measure transform. Let $\varphi \in S-d(\R)$ and $B$ a real bounded operator on $S_d(\R)'$. Then we have
$$
\int_{S(\R)'} \sigma_{B} \varphi(\omega) \, d\mu(\omega)
=\int_{S(\R)'} \varphi(B\omega) \, d\mu(\omega)
=\int_{S(\R)'} \varphi(\omega) \, d\mu(B^{-1} \omega)
$$
Moreover we have 
$$\int_{S(\R)'} \exp(i\langle \xi, \omega) \, d\mu(B^{-1} \omega)
= \exp(-\frac{1}{2} \langle B^* \xi, B^* \xi \rangle,$$
which is a characteristic function of a probability measure by the Theorem of Bochner and Minlos \cite{Mi63}.
Furthermore 
$$\int_{S(\R)'} \exp(i\langle \xi, \omega) \, d\mu(B^{-1} \omega)= T(\sigma^{\dag} \1)(\xi),$$
such that $\Phi_{BB^*}$ is represented by the positive Hida measure $\mu\circ B^{-1}$.
\end{remark}

\section{Construction of Hamiltonian Path Integrand via Generalized Scaling}

We construct in this section by a suitable generalized scaling the Hamiltonian path integral as an expectation based on the formula above as in \eqref{thecomplex}.\\
In phase space however the arguments are multidimensional, since we consider momentum and space variables as independent variables. For simplicity we consider $t_0=0$ and futhermore $\hbar =m =1$.
Indeed we have the following

\begin{proposition}
Let $N^{-1} = \bigg(
\begin{array}{l l}
 \1_{[0,t)^c} & 0\\
    0 & \1_{[0,t)^c}
\end{array} 
\bigg)+i\bigg(
\begin{array}{l l}
  \1_{[0,t)}& \1_{[0,t)}\\
    \1_{[0,t)}& 0
\end{array} 
\bigg)$ as in the case of the free Hamiltonian integrand $I_0$ (i.e. $V=0$). Let $R$ be a symmetric operator (w.r.t.~the dual pairing) with $R^2=N^{-1}$. Indeed we have:
$$
R= \bigg(
\begin{array}{l l}
 \1_{[0,t)^c} & 0\\
    0 & \1_{[0,t)^c}
\end{array} 
\bigg)+
\frac{\sqrt{i}}{1+(\frac{\sqrt{5}+1}{2})^2} U^T \left(\begin{array}{c c} \frac{1+\sqrt{5}}{2} \1_{[0,t)} & 0\\
0& \frac{1-\sqrt{5}}{2} \1_{[0,t)} \end{array}\right) U,
$$
with $$U=\left(\begin{array}{c c} -\frac{\sqrt{5}+1}{2} & 1 \\ -1 & -\frac{\sqrt{5}+1}{2} \end{array}\right)$$
Then under the assumption that $\sigma_R \delta( \langle (\1_{[0,t)},0),\cdot\rangle) = \delta( \langle R(\1_{[0,t)},0),\cdot\rangle)\in (S)'$, we have 
$$ I_0 =\sigma^{\dag}_R \sigma_R  \delta( \langle (\1_{[0,t)},0),\cdot\rangle).$$
\end{proposition} 

\noindent Consequently the Hamiltonian path integrand for an arbitrary space dependent potential $V$, can be informally written as
\begin{multline}\label{scIV}
I_V = \mathrm{Nexp} \left(-\frac{1}{2} \langle \cdot,K \cdot \rangle \right) \exp\left(-i \int_0^t V(x_0+\langle (\1_{[0,r)},0),\cdot \rangle)\,dr\right)\\
\times\delta(x_0+\langle (\1_{[t_0,t)},0),\cdot \rangle-x)\\
=\sigma^{\dag}_R \Big( \sigma_R \Big(\exp\left(-i \int_0^t V(x_0+\langle (\1_{[0,r)},0),\cdot \rangle)\,dr\right)\Big)\\
\times\sigma_R\delta(x_0+\langle (\1_{[t_0,t)},0),\cdot \rangle-x)\Big),
\end{multline}
for $x,x_0 \in \R$ and $0<t_0<t<\infty$.

In the following we give some ideas to give a mathematical meaning to the expression in \eqref{scIV}.
First we consider a quadratic potential, i.e.~we consider 
$$\exp(-\frac{1}{2} \langle \cdot L \cdot \rangle) \delta(\langle (\1_{[t_0,t)},0),\cdot \rangle-x).$$
\begin{definition}
For $L$ fulfilling the assumption of Lemma \ref{thelemma} and $$\delta(\langle (\1_{[t_0,t)},0),\cdot \rangle-x)$$ we define
\begin{multline*}
\sigma_{R}\left( \exp(-\frac{1}{2} \langle \cdot L \cdot \rangle) \delta(\langle (\1_{[t_0,t)},0),\cdot \rangle-x) \right) \\:=\exp(-\frac{1}{2} \langle \cdot RLR \cdot \rangle) \delta(\langle R(\1_{[t_0,t)},0),\cdot \rangle-x).
\end{multline*}
\end{definition}
We now take a look at the $T$-transform of this expression in $\bmf \in S_2(\R)$. We have
\begin{multline*}
T(\sigma_R^{\dag} \sigma_R \left(\exp(-\frac{1}{2} \langle \cdot L \cdot \rangle) \delta(\langle (\1_{[t_0,t)},0),\cdot \rangle-y)\right))(\bmf)\\
=T(\sigma_R \left(\exp(-\frac{1}{2} \langle \cdot L \cdot \rangle) \delta(\langle (\1_{[t_0,t)},0),\cdot \rangle-y)\right))(R\bmf)\\
= \frac{1}{\sqrt{2 \pi \det(Id+RLR)}} \exp(-\frac{1}{2} \langle R\bmf, (Id+RLR)^{-1} R\bmf \rangle)\\ \times\exp\big(\frac{(iy -\langle R\bmf,(Id+RLR)^{-1} R(\1_{[t_0,t)},0)\rangle)^2}{2 \langle R(\1_{[t_0,t)},0),(Id+RLR)^{-1} R(\1_{[t_0,t)},0)\rangle } \big).
\end{multline*}
Now with $R^2=N^{-1}$ and since $R$ is invertible with $R^{-1} R^{-1}= N$, we have
$$Id+RLR = R R^{-1} R^{-1} R + RLR = R(Id +K+L)R$$
and 
$$
(Id +RLR)^{-1} = R^{-1}(Id +K+L)^{-1} R^{-1}.
$$
Thus 
\begin{multline*}
T(\sigma_R^{\dag} \sigma_R \left(\exp(-\frac{1}{2} \langle \cdot L \cdot \rangle) \delta(\langle (\1_{[t_0,t)},0),\cdot \rangle-y)\right))(\xi)\\
= \frac{1}{\sqrt{2 \pi \det((N+L)N^{-1})}} \exp(-\frac{1}{2} \langle \xi, (N+L)^{-1} \xi \rangle)\\ \exp\big(\frac{1}{2 \langle (\1_{[t_0,t)},0),(N+L)^{-1} (\1_{[t_0,t)},0)\rangle } (iy -\langle \xi,(N+L)^{-1} (\1_{[t_0,t)},0)\rangle)^2 \big),
\end{multline*}
which equals the expression from Lemma \ref{thelemma}. Hence we have that for a suitable quadratic potential 
$$\sigma_R^{\dag} \sigma_R \left(\exp(-\frac{1}{2} \langle \cdot L \cdot \rangle) \delta(\langle (\1_{[t_0,t)},0),\cdot \rangle-y)\right),$$
exists as a Hida distribution. Moreover for all quadratic potentials from the previous chapter, the $T$-transform obtained via scaling gives the generating functional as in chapter 8. Since the $T$-transforms coincide, also the distributions are the same.\\
For the case of quadratic potentials we obtained the correct physics also by the scaling approach. 

\begin{example}
We construct the Feynman integrand for the harmonic oscillator in phase space via the generalized scaling. I.e.~the potential is given by $x \mapsto V(x)= \frac{1}{2}k x^2$, $0 \leq k<\infty$.\\
Thus the matrix $L$ which includes the information about the potential, is given by
\begin{equation*}
L=\left(
\begin{array}{l l}
i k A & 0\\
0 & 0
\end{array}
\right),\, y \in \R, \,t>0,
\end{equation*}
where $A \,f(s)=\1_{[0,t)}(s) \int_s^t \int_0^{\tau} f(r) \, dr \, d\tau, f \in L^2(\R,\C), s\in \R$, then for $\bmf \in S_2(\R)_{\C}$, see also \cite{BG11} and \cite{GS98a} we have
\begin{multline*}
T\left(\sigma_R^{\dag} \sigma_R \left(\exp(-\frac{1}{2} \langle \cdot L \cdot \rangle) \delta(\langle (\1_{[t_0,t)},0),\cdot \rangle-y)\right)\right)(\bmf)\\ =\sqrt{\left(\frac{\sqrt{k}}{2\pi i \sin(\sqrt{k} t)}\right)} \exp\!\left( \frac{1}{2} \frac{\sqrt{k}}{i\tan(\sqrt{k} t)} \Big(iy+\big({\boldsymbol{\eta}}, {\bf f} +{\bf g}\big) \Big)^2\right)\\
\times\exp\!\Bigg(-\frac{1}{2} \bigg( \big({\bf f} + {\bf g}\big) ,\! 
 \left(
\begin{array}{l l}
\1_{[0,t)^c} & 0 \\
0 & \1_{[0,t)^c} \end{array}\right)
\big({\bf f} + {\bf g}\big) \bigg)\!\Bigg)\\
\times\exp\!\Bigg(-\frac{1}{2} \bigg( \big({\bf f} + {\bf g}\big) ,\! 
 \frac{t}{i}
\1_{[0,t)} \left(
\begin{array}{l l}
\frac{1}{t}(kA- \1_{[0,t)})^{-1} &  (kA- \1_{[0,t)})^{-1} \\
(kA- \1_{[0,t)})^{-1}& ktA(kA- \1_{[0,t)})^{-1} \end{array}\right)
\big({\bf f} + {\bf g}\big) \bigg)\!\Bigg),
\end{multline*}
which is identically equal to the generating functional of the Feynman integrand for the harmonic oscillator in phase space, see e.g.\cite{BG11}.\\
Moreover its generalized expectation
\begin{equation*}
\mathbb{E}(I_{HO})=T(I_{HO})(0)=\sqrt{\left(\frac{\sqrt{k}}{2\pi i \sin(\sqrt{k} t)}\right)} \exp\left( i \frac{\sqrt{k}}{2\tan(\sqrt{k} t)} y^2\right)
\end{equation*}
is the Greens function to the Schrö\-dinger equation for the harmo\-nic oscil\-lator, compare e.g.~with \cite{KL85}.   
\end{example}
\section*{Acknowledgements}
Dear Ludwig Streit, I wish you all the best to your 75th birthday. The author would like to thank the organizing and programme committee of the Stochastic and Infinite Dimensional Analysis conference for an interesting an stimulating meeting.


\begin{thebibliography}{1}
\bibitem{AGM02}
S.~Albeverio, G.~Guatteri, and S.~Mazzucchi.
\newblock Phase space {Feynman} path integrals.
\newblock 43(6):2847--2857, June 2002.

\bibitem{AHKM08}
S.~Albeverio, R.~{H{\o}egh-Krohn}, and S.~Mazzucchi.
\newblock {\em Mathematical Theory of {F}eynman Path Integrals: An
  Introduction}, volume 523 of {\em Lecture Notes in Mathematics}.
\newblock Springer Verlag, Berlin, Heidelberg, New York, 2008.

\bibitem{B13}
W.~Bock.
\newblock Hamiltonian path integrals in momentum space representation via white
  noise techniques.
\newblock {\em Rep. Mat. Phys.}, 2013.
\newblock accepted for publication.

\bibitem{BG11}
W.~Bock and M.~Grothaus.
\newblock {A White Noise Approach to Phase Space Feynman Path Integrals}.
\newblock {\em Teor. Imovir. ta Matem. Statyst.}, (85):7--21, 2011.

\bibitem{BG13}
W.~Bock and M.~Grothaus
\newblock The hamiltonian path integrand for the charged particle in a constant
  magnetic field as white noise distribution, 2013.

\bibitem{BK95}
Y.~M. Berezansky and Y.~G. Kondratiev.
\newblock {\em Spectral methods in infinite-dimensional analysis. {V}ol. 2},
  volume 12/2 of {\em Mathematical Physics and Applied Mathematics}.
\newblock Kluwer Academic Publishers, Dordrecht, 1995.
\newblock Translated from the 1988 Russian original by P. V. Malyshev and D. V.
  Malyshev and revised by the authors.

\bibitem{DMN77}
C.~DeWitt-Morette, A.~Maheshwari, and B.~Nelson.
\newblock Path integration in phase space.
\newblock {\em General Relativity and Gravitation}, 8(8):581--593, 1977.

\bibitem{D80}
H.~{Doss}.
\newblock {Sur une Resolution Stochastique de l'Equation de Schr{\"o}dinger
  {\`a} Coefficients Analytiques}.
\newblock {\em Communications in Mathematical Physics}, 73:247--264, October
  1980.

\bibitem{F48}
R.~P. Feynman.
\newblock Space-time approach to non-relativistic quantum mechanics.
\newblock {\em Rev. Modern Physics}, 20:367--387, 1948.

\bibitem{Fe51}
R.P. Feynman.
\newblock An operator calculus having applications in quantum electrodynamics.
\newblock {\em Physical Review}, 84(1):108--124, 1951.

\bibitem{FeHi65}
R.P. Feynman and A.R. Hibbs.
\newblock {\em Quantum Mechanics and Path Integrals}.
\newblock McGraw-Hill, London, New York, 1965.

\bibitem{GS98a}
M.~Grothaus and L.~Streit.
\newblock Quadratic actions, semi-classical approximation, and delta sequences
  in {G}aussian analysis.
\newblock {\em Rep. Math. Phys.}, 44(3):381--405, 1999.

\bibitem{GSV08}
M.~Grothaus, L.~Streit, and A.~Vogel.
\newblock The feynman integrand as a white noise distribution for
  non-perturbative potentials.
\newblock 2008.
\newblock to appear in the festschrift in honor of Jean-Michel Bismut's
  sixtieth birthday.

\bibitem{GSV10}
M.~Grothaus, L.~Streit, and A.~Vogel.
\newblock The complex scaled feynman--kac formula for singular initial
  distributions.
\newblock {\em Stochastics}, 84(2-3):347--366, April-June 2012.

\bibitem{GV68}
I.M. Gel'fand and N.Ya. Vilenkin.
\newblock {\em Generalized Functions}, volume~4.
\newblock Academic Press, New York, London, 1968.

\bibitem{GV08}
M.~Grothaus and A.~Vogel.
\newblock The {F}eynman integrand as a white noise distribution beyond
  perturbation theory.
\newblock To appear in the proceedings of the "5th Jagna International
  Workshop: Stochastic and Quantum Dynamics of Biomolecular Systems", 2008.

\bibitem{Hid80}
T.~Hida.
\newblock {\em Brownian motion}, volume~11 of {\em Applications of
  Mathematics}.
\newblock Springer-Verlag, New York, 1980.
\newblock Translated from the Japanese by the author and T. P. Speed.

\bibitem{HKPS93}
T.~Hida, H.-H. Kuo, J.~Potthoff, and L.~Streit.
\newblock {\em White Noise. An infinite dimensional calculus}.
\newblock Kluwer Academic Publisher, Dordrecht, Boston, London, 1993.

\bibitem{HS83}
T.~Hida and L.~Streit.
\newblock Generalized brownian functionals and the feynman integral.
\newblock {\em Stoch. Proc. Appl.}, 16:55--69, 1983.

\bibitem{KD82}
J.~R. Klauder and I.~Daubechies.
\newblock Measures for path integrals.
\newblock {\em Physical Review Letters}, 48(3):117--120, 1982.

\bibitem{KD84}
J.~R. Klauder and I.~Daubechies.
\newblock Quantum mechanical path integrals with wiener measures for all
  polynomial hamiltonians.
\newblock {\em Physical Review Letters}, 52(14):1161--1164, 1984.

\bibitem{Ku11}
Naoto Kumano-Go.
\newblock {Phase space Feynman path integrals with smooth functional
  derivatives by time slicing approximation.}
\newblock {\em Bull. Sci. Math.}, 135(8):936--987, 2011.

\bibitem{KL85}
D.C. Khandekar and S.V. Lawande.
\newblock Feynman path integrals: Some exact results and applications.
\newblock {\em Physics Reports}, 137(2), 1986.

\bibitem{KLPSW96}
Yu.G. Kondratiev, P.~Leukert, J.~Potthoff, L.~Streit, and W.~Westerkamp.
\newblock Generalized functionals in {G}aussian spaces: The characterization
  theorem revisited.
\newblock {\em J. Funct. Anal.}, 141(2):301--318, 1996.

\bibitem{Kon80}
Yu.G. Kondratiev.
\newblock Spaces of entire functions of an infinite number of variables,
  connected with the rigging of a {F}ock space.
\newblock {\em Selecta Mathematica Sovietica}, 10(2):165--180, 1991.
\newblock Originally published in Spectral Analysis of Differential Operators,
  Mathematical Institute of the Academy of Sciences of the Ukrainian SSR, Kiev,
  1980, pp. 18--37.

\bibitem{Kuo96}
H.-H. Kuo.
\newblock {\em White Noise Distribution Theory}.
\newblock CRC Press, Boca Raton, New York, London, Tokyo, 1996.

\bibitem{LLSW94}
A.~Lascheck, P.~Leukert, L.~Streit, and W.~Westerkamp.
\newblock More about {D}onsker's delta function.
\newblock {\em Soochow J. Math.}, 20(3):401--418, 1994.

\bibitem{Mi63}
R.~A. Minlos.
\newblock Generalized random processes and their extension to a measure.
\newblock {\em Selected Transl. Math. Statist. and Prob.}, 3:291--313, 1963.

\bibitem{Ob94}
N.~Obata.
\newblock {\em White Noise Calculus and Fock Spaces}, volume 1577 of LNM.
\newblock Springer Verlag, Berlin, Heidelberg, New York, 1994.

\bibitem{PS91}
J.~Potthoff and L.~Streit.
\newblock A characterization of {H}ida distributions.
\newblock {\em J. Funct. Anal.}, 101:212--229, 1991.

\bibitem{RS75a}
M.~Reed and B.~Simon.
\newblock {\em Methods of modern mathematical physics}, volume~I.
\newblock Academic Press, New York, London, 1975.

\bibitem{Ru74}
Walter Rudin.
\newblock {\em Functional analysis}.
\newblock International Series in Pure and Applied Mathematics. McGraw-Hill
  Inc., New York, first edition, 1974.

\bibitem{V10}
A.~Vogel.
\newblock A new wick formula for products of white noise distributions and
  application to feynman path integrands, 2010.


\bibitem{W95}
W.~Westerkamp.
\newblock {\em Recent Results in Infinite Dimensional Analysis and Applications
  to {F}eynman Integrals}.
\newblock PhD thesis, University of Bielefeld, 1995.



\end{thebibliography}
\end{document}